\documentclass[reqno,11pt]{article} 

\usepackage{amsmath,amsthm,amssymb,amscd,amstext,amsfonts}
\usepackage{mathtools}
\usepackage{mathrsfs}
\usepackage{thmtools}
\usepackage{latexsym}
\usepackage{verbatim}
\usepackage{framed}
\usepackage{graphicx}
\usepackage{stmaryrd}
\usepackage{enumerate}
\usepackage{fullpage}
\usepackage{bm}
\usepackage{color}
\usepackage{hyperref}
\usepackage{url}
\usepackage{physics}
\usepackage{multicol}
\usepackage{multirow}
\usepackage{tikz}
\usepackage{tikz-cd}
\usetikzlibrary{decorations.pathmorphing}   
\usetikzlibrary{arrows}
\usepackage{makeidx}
\usepackage{makecell}		
\usepackage{algorithm}
\usepackage[noend]{algpseudocode}
\usepackage{nicematrix}   
\usepackage{authblk}  
\usepackage{cite}

\newtheorem{theorem}{Theorem}

\newtheorem{corollary}[theorem]{Corollary}

\DeclareMathOperator{\FF}{FirstFit}
\DeclareMathOperator{\CBIP}{CBIP}

\newcommand{\N}{\mathbb{N}}

\begin{document}

\title{$\FF$ online coloring in the random order model}

\author{
    Xinyu Ye\footnote{Email: yxyu613@163.com}, 
    Yuechuan Xu,
    Zixuan Wang, 
    Jiaying Zheng,
    Yaqiao Li\footnote{corresponding author: liyaqiao@suat-sz.edu.cn}\\
    Shenzhen University of Advanced Technology, China
}

\maketitle

\begin{abstract}
    The average performance of FirstFit online coloring on trees in the random order model is completely determined in recent works of Frei et al  and Bosek et al, showing $\Theta(\log n /\log\log n)$ number of colors, improving the $\Theta(\log n)$ colors in the adversarial model. We provide a few further results on slightly more general graph classes. Firstly, We extend their method to obtain a simple path-counting principle for sparse graph classes, which immediately yields for example that cactus graphs and uniform hypertrees exhibit a similar improvement. We then show that FirstFit uses only $O(1)$ colors on crown graphs, a standard example where adversarial arrival forces  $\Theta(n)$  colors. We further show that density alone (even linear minimum degree) is insufficient to guarantee $O(1)$ colors even on bipartite graphs. Finally, we identify graph classes, including unit interval graphs and some graphs of high chromatic number, for which random arrival provides only limited improvement. We end with some open problems.
\end{abstract}

\section{Introduction} \label{sec:intro}

This paper studies the vertex arrival model of online graph vertex coloring. In the adversarial model,  an adversary creates a simple undirected graph 
$G=(V,E)$ of $n$ vertices, and selects a presentation order of vertices $\sigma : [n] \rightarrow V$. The graph is then presented online: at time $i$, vertex $v=\sigma(i)$ arrives, and we learn all its neighbors among already appeared vertices.
We call the set of neighbors of $v$ that an online algorithm learns about at the time of arrival of $v$ as the pre-neighborhood of $v$. 
An online algorithm must declare how to irrevocably color the new vertex $v=\sigma(i)$ prior to the arrival of the next vertex $\sigma(i+1)$. An online algorithm does not know $G$ and even $n$. Alternatively, we can view the online input as a sequence of induced subgraphs:
    $G \cap \sigma([1]), G \cap \sigma([2]), \ldots, G \cap \sigma([n])$.
Given a graph $G$, in the adversarial model, the performance of an online algorithm is measured by its maximum number of colors used among all possible presentation orders. In other words, the adversarial model is for worst case analysis.

An extensively studied online coloring algorithm is a greedy algorithm which is called $\FF$ in the literature: when a vertex $v$ arrives, $\FF$ colors it with the first color that does not appear in its  pre-neighborhood. A sequence of works  \cite{graph_gyarfas1988line, irani1990coloring, albers2021tight} showed that $\FF$ is optimal on trees, inductive graphs, and chordal graphs, etc, using $\Theta(\chi(G) \cdot \log n)$ colors, where $\chi(G)$ is the chromatic number of the graph $G$. Furthermore, it is shown that $\FF$ is also optimal on uniform hypertrees, see \cite{hyperC_1,li2023online}.

For practical purpose, it  would be  desirable to understand the average performance of algorithms. For online coloring, one way of obtaining such understanding is to analyze the \emph{random order model} (aka, random arrival model), in which the adversary still selects the graph, \emph{but the presentation order $\sigma$ is chosen uniformly at random}. While there are many studies on the adversarial model, the random order model is much less understood. Recently, Frei et al \cite{ROM_treecoloring_prediction} and Bosek et al \cite{ROM_FF} studied $\FF$ in the random order model, they showed that on trees (forests), $\FF$ uses $\Theta(\log n /\log\log n)$ colors in expectation, demonstrating that the average performance of $\FF$ is slightly better than the worst case on trees.

This work makes modest contribution in providing several further results in the random order model, demonstrating several  phenomena of some interest. In Section \ref{sec:k-path} and Section \ref{sec:dense}, we exhibit several other graph classes that shows $\FF$ has better, sometimes nearly optimal, performance in the random order model. We extend the  method of Frei et al \cite{ROM_treecoloring_prediction} and Bosek et al \cite{ROM_FF}  to obtain a simple path-counting principle for sparse graph classes, which immediately yields $O(\log n /\log\log n)$ colors for Cactus graphs and uniform hypertrees. However, we point out the limitation of this path-counting technique by providing a simple construction, belonging to several graph classes of interest, on which it fails. 
We then show that FirstFit uses only $O(1)$ colors on crown graphs, a standard example where adversarial arrivals force  $\Theta(n)$  colors.  On the other hand, we show that even in very dense bipartite graphs, $\FF$ still uses $\Omega(\log n /\log\log n)$ colors even in the random order model, i.e.,  being dense alone cannot guarantee $O(1)$ colors for $\FF$ in the random order model. We point out that since some parts of proofs in Section \ref{sec:k-path} are essentially the same as in Frei et al \cite{ROM_treecoloring_prediction} and Bosek et al \cite{ROM_FF}, to avoid repetition, we direct the interested readers to their papers for details.
Lastly, in Section \ref{sec:less_adv}, we show that on certain graphs $\FF$ in the random order model has only limited benefit comparing to the adversarial model. 

Throughout the paper, $n$ always denotes the number of vertices of a graph in the corresponding context, unless explicitly specified otherwise. 

\medskip

\noindent {\bf Acknowledgment.} Y.L. thanks  Vishnu V. Narayan and Denis Pankratov for useful discussion.

\section{Bounded $k$-paths graphs}  \label{sec:k-path}

In this section, we slightly generalize Frei et al \cite{ROM_treecoloring_prediction} and Bosek et al \cite{ROM_FF}, and provide, as examples, two other sparse graph classes   on which $\FF$ also uses $\Theta(\log n /\log\log n)$ colors in expectation in the random order model.

Given a simple graph $G$, a parameter $k \in \N$, a $k$-path in $G$ consists of $k$ vertices $v_1, \ldots, v_k \in V(G)$ together with edges $e_{i,i+1} \in E(G)$ that connects $v_i$ and $v_{i+1}$ for every $i=1, \ldots, k-1$. 
Two $k$-paths are distinct if they differ at least one edge. Let $P_k(G)$ denote the number of distinct $k$-paths in $G$. 

\begin{theorem} \label{thm:path-counting}
    Let $G$ be a simple graph with $n$ vertices. If there exist constants $a,b >0$ such that $P_k(G) \le O(n^a b^k)$ for every relevant $k$, then in the random order model, $\FF$ uses at most $O(\log n /\log\log n)$ colors in expectation on $G$. 
\end{theorem}

\begin{proof}
    The proof is essentially the same as in \cite{ROM_treecoloring_prediction,ROM_FF}. We briefly sketch its main steps below. If a vertex $v$ in $G$ is colored by $k$ by $\FF$, then necessary conditions are that  there exists a $k$-path $u_1\ldots u_{k-1} v$ that ends at $v$, and such that $\FF$ assigns color $i$ to vertex $u_i$, and these $k$ vertices appear in exactly the \emph{relative} order as $u_1, \ldots, u_{k-1}, v$. Since vertices appear in uniform random order,  the probability that the $k$ vertices $u_1, \ldots, u_{k-1}, v$ appear in exactly this relative order is precisely $1/k!$. Let $E$ denote the probabilistic event that in a random order, there  exists a $k$-path $u_1\ldots u_{k-1} v$ that ends at $v$  and the $k$ vertices appear in exactly that relative order. Then, 
    \[
        \Pr[\exists\ v, \FF \text{ uses color $k$ on $v$}]
        \le \Pr[E] \le \frac{2P_k(G)}{k!},
    \]
    where the number $2$ denotes that $v$ can be chosen as one of the two ending vertices of a $k$-path. 
    Hence, the expected number of colors on $G$ used by $\FF$ is at most
    \begin{align*}
        (k-1) + n \cdot \Pr[\FF \text{ uses at least $k$ colors}] 
        &\le (k-1)  + n \cdot \Pr[\exists\ v, \FF \text{ uses color $k$ on $v$}]  \\
        &\le (k-1) + n \cdot \frac{2P_k(G)}{k!} \\
        &\le (k-1) + n \cdot \frac{O(n^a b^k)}{k!}.
    \end{align*}
    By optimizing with $k$, a simple calculation shows that this is $O(\log n / \log\log n)$ as desired.
\end{proof}

The $O(\log n/\log\log n)$ upper bound for trees follows from the observation that  if $G$ is a tree, then $P_k(G) \le O(n^2)$. We can similarly obtain this lower bound for some other graph classes, two examples are given here. 
\begin{itemize}
    \item Cactus graphs. In a Cactus graph, each edge is contained in at most one cycle. The definition immediately implies that if $G$ is a cactus graph, then $P_k(G) \le O(n^2 2^k)$, because you have only 2 choices along each cycle when forming a path. 
    \item Hypertrees. In a hypergraph, a Berge path of length $k$ is an alternating sequence of vertices and hyperedges: $v_1, e_1, v_2, e_2, \dots, v_k$ such that all vertices $v_i$ are distinct and all hyperedges  $e_i$  are distinct, and such that for each $i \in \{1, \dots, k-1\}$, the hyperedge $e_i$ contains both $v_i$ and $v_{i+1}$.  A Berge cycle can be similarly defined except that the last vertex coincides with the first vertex. A hypergraph is called a (Berge) hypertree if it is connected and does not contain any Berge cycle. Because of this (Berge) cycle-free property, the number of $k$-paths (in terms of Berge path) in a hypertree is $O(n^2)$ just like the standard trees. A hypergraph is called $s$-uniform if its every hyperedge contains exactly $s$ vertices. Hypergraph coloring is defined similarly as graph coloring: requiring that no hyperedge is monochromatic.
\end{itemize}

Theorem \ref{thm:path-counting} then implies the following. 

\begin{corollary}   \label{cor:cactus_hypertree}
    For  cactus graphs,   $\FF$ uses in expectation $\Theta(\log n /\log\log n)$ colors in the random order model. For $s$-uniform hypertrees,  $\FF$ uses in expectation $\Theta_s(\log n /\log\log n)$ colors in the random order model. 
\end{corollary}

\begin{proof}
    The upper bounds follow directly from Theorem  \ref{thm:path-counting}  and the bounds on $k$-paths as given above. Since trees is a subset of cactus graphs, the lower bound $\Omega(\log n /\log\log n)$ established by Bosek et al \cite{ROM_FF} also holds for cactus graphs. 
    
    Furthermore, observe that the lower bound proof in \cite{ROM_FF} can be directly adapted to work for $s$-uniform hypertrees. We only point out the construction of the rooted hypertree which is a direct generalization of \cite{ROM_FF}(see their paper Section 3, Figure 1). Let $H^r_1$ be a single vertex. Construct the rooted hypertree $H^r_i$ recursively as follows: make $r(s-1)$ copies of  $H^r_{j}$ for every $1 \le j \le i-1$ , view the $r(s-1)$ copies of  $H^r_{j}$ as $r$ groups, so each group has size $s-1$. Create a single new vertex $v$, and connect $v$ to the $s-1$ root vertices of each group  to form an $s$-uniform hyperedge. Let $v$ be the root vertex of $H^r_i$.
    The rest of the proof works in the same way. We refer interested readers to \cite{ROM_FF} for details. The number of vertices of $H^r_{k}$ is $n = O((sr)^k)$ instead of $O(r^k)$. Choosing $r= \Theta(k\log k)$, the proof in \cite{ROM_FF}  shows the probability that $\FF$ uses color $k$ on $H^r_{k}$ has probability at least $\Omega(1)$, which implies the lower bound $\Omega_s(\log n /\log\log n)$.
\end{proof}

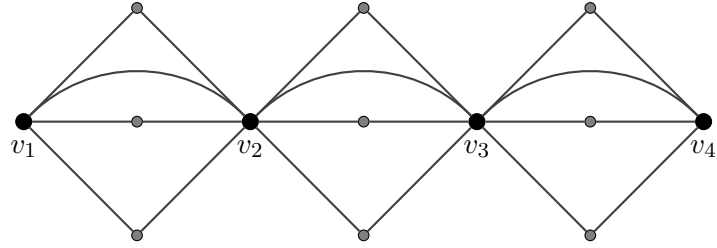
\begin{figure}[htbp!]
    \centering
\begin{tikzpicture}[
    cut-node/.style={circle, draw=black, fill=black, inner sep=0pt, minimum size=6pt},
    mid-node/.style={circle, draw=black, fill=gray, inner sep=0pt, minimum size=4pt},
    edge/.style={thick, darkgray}
]

\node[cut-node, label=below:{$v_1$}] (v1) at (0, 0) {};
\node[cut-node, label=below:{$v_2$}] (v2) at (3, 0) {};
\node[cut-node, label=below:{$v_3$}] (v3) at (6, 0) {};
\node[cut-node, label=below:{$v_4$}] (v4) at (9, 0) {};

\foreach \i [evaluate=\i as \nexti using int(\i+1)] in {1, 2, 3} {
    
    \draw[edge] (v\i) to[bend left=45] (v\nexti);

    \foreach \y in {1.5, 0, -1.5} {
        \node[mid-node] (u) at (\i*3 - 1.5, \y) {};
        
        \draw[edge] (v\i) -- (u) -- (v\nexti);
    }
}
\end{tikzpicture}
    \caption{The graph $G_{4,3}$, it is series–parallel, of treewidth 2, is 2-inductive and planar.} \label{fig:unbounded_k_path}
\end{figure}

Unfortunately, inductive graphs and bounded treewidth graphs do not necessarily satisfy the bounded $k$-paths property in Theorem \ref{thm:path-counting}. A simple example is illustrated in Figure \ref{fig:unbounded_k_path}, where the graph $G_{4,3}$ can be directly generalized to $G_{p,q}$ for any $p, q$: there are $p$ vertices $v_1, \ldots, v_p$ forming a path, and between every two consecutive $v_i$ and $v_{i+1}$, one adds $q$ vertices and connect all of them to both $v_i$ and $v_{i+1}$. One can directly check that the graph $G_{p,q}$:
\begin{itemize}
    \item is series–parallel,
    \item has treewidth 2, 
    \item is 2-inductive, 
    \item is planar.
\end{itemize}
The number of vertices of $G_{p,q}$ is $n= p+q(p-1) = \Theta(pq)$. However, $P_k(G_{p,q}) \ge \Omega(q^k)$. Choosing $p=q$, we would have $n= \Theta(p^2)$ and $P_k(G_{p,q}) \ge \Omega(n^{k/2})$ which grows too fast  to apply Theorem \ref{thm:path-counting}.

\section{Dense bipartite graphs}    \label{sec:dense}

Theorem \ref{thm:path-counting} mainly applies to sparse graph classes. We next show that sparsity is not the only mechanism leading to good random-order behavior. Intuitively, dense bipartite graphs seem also easy for $\FF$ because of its high connectivity. Below we exhibit two such classes with different phenomena.

\subsection{The crown graph}    \label{sec:crown}

A crown graph on $2n$ vertices, call it $C=(U,V;E)$, is an undirected graph with two sets of vertices 
    $U= \{u_1, u_2, \ldots, u_n\}$ 
and 
    $V= \{v_1, v_2, \ldots, v_n\}$ 
and with an edge from $u_i$ to $v_j$ whenever 
    $i \neq j$.
In the adversarial online graph coloring model, the crown graph is the typical graph that, when the vertices are given in the order 
    $u_1, v_1, u_2, v_2, \ldots$, 
forces $\FF$ to use $n$ colors. In other words, the crown graph is a witness for the poorest performance of $\FF$ in the worst case scenario. It is then a natural question to understand how $\FF$ performs in the random order model.

\begin{theorem} \label{thm:crown_graph}
    Let $C=(U,V;E)$ be a crown graph on $2n$  vertices, where $|U|=|V|=n$. Then, in the random order model $\FF$ uses at most $2.5 + O(1/n)$ colors in expectation. 
\end{theorem}

\begin{proof}
    For every 
        $u \in U$,
    let 
        $\overline{u} \in V$
    denote the unique vertex in $V$ for which there is no edge connecting $u$ and $\overline{u}$. 
    Let 
        $\sigma$
    be a permutation of $V(C) = U\cup V$.
    Partition $\sigma$ into consecutive blocks, where each block is a subset of $U$ or a subset of $V$. In this way, we write
        \[
            \sigma = T_1 \circ T_2 \circ T_3 \cdots, 
            \quad
            T_1, T_3, \cdots \subseteq U, 
            \quad
            T_2, T_4, \cdots \subseteq V. 
        \]
    Let
        $t_i = |T_i|$.
    We discuss below the case $T_1 \subseteq U$. The case  $T_1 \subseteq V$ is similar. 
    By abuse of notation we use $\FF$ to denote the number of colors used.
    \begin{enumerate}[(1)]
        \item $t_1 \ge 2$. This implies $\FF = 2$. 
        
        This is because in this case every vertex in $T_2$ will get color $2$, and every vertex in $T_3$ will get color $1$, etc. 
        
        \item $t_1 = 1, t_2 \ge 3$. This implies $\FF \le 3$.
        Let $x$ be the unique vertex in $T_1$. 
        If $\overline{x}\notin T_2$, then every vertex in $T_2$ receives color $2$, and the subsequent blocks alternate between colors $1$ and $2$. If $\overline{x}\in T_2$, then $\overline{x}$ receives color $1$ and the other vertices in $T_2$ receive color $2$. Since $t_2\ge3$, every subsequent vertex in $U$ sees both colors $1$ and $2$ and receives color $3$, while every subsequent vertex in $V$ receives color $2$. Thus, $\FF$ uses $3$ colors if and only if $\overline{x}\in T_2$. This requires $\overline{x}$ to appear before every vertex in $U\setminus\{x\}$, so this occurs with probability at most $1/n=O(1/n)$.
        
        \item $t_1 = 1, t_2 \le 2$, and $\overline{x} \not\in T_2$, where $x$ is the unique vertex in $T_1$. This implies $\FF = 2$.
        
        In this case, $T_2$ receives only color $2$ and $T_3$ only color $1$; the same alternation continues thereafter.

        \item $t_1 = 1, t_2 \le 2$, and $\overline{x} \in T_2$, where $x$ is the unique vertex in $T_1$. 

        Let $E_4$ denote the event that this case happens. We have,
        \begin{align*}
            \Pr[E_4] 
            &\le 
            \Pr[t_1=1, \text{ and $\overline{x}$ appears in the first two vertices of $T_2$}] \\
            &= \frac{2(n-1)}{(2n-1)(2n-2)} 
            = \frac{1}{2n-1}.
        \end{align*}
        Also, $\FF$ uses at most $n$ colors because the graph has maximum degree 
            $\Delta(C) = n-1$. 
    \end{enumerate}

    To sum up, the number of colors used is either 2, or 3, or at most $n$. Hence, the expected number of colors of $\FF$ is
        $\le 2 + 3 \cdot O(\frac{1}{n}) + n \cdot \frac{1}{2n-1} \le 2.5 + O(\frac{1}{n})$.
\end{proof}

\subsection{Linear minimum degree}  \label{sec:min_deg}

On the other hand, below we show that edge density, even linear minimum degree, is insufficient to guarantee a $O(1)$ upper bound for $\FF$. The lower bound construction is a suitable blow up of the crown graph.

\begin{theorem} \label{thm:linear_min_degree}
    There exists a bipartite graph of $n$ vertices with minimum degree $\Omega(n)$ such that $\FF$ uses $\Omega(\log n/\log\log n)$ colors in expectation in the random order model.
\end{theorem}

\begin{proof}
    We will construct a desired graph with vertices $n=O(k^k)$ and minimum degree $\Omega(n)$, such that with $\Omega(1)$ probability $\FF$ uses at least $k$ colors. 
    
    Let $G_k$ be constructed as follows, which is a suitable blow up of the crown graph. Let $H_k = \{A_k^{(1)}$, $A_k^{(2)}$, $\ldots$, $A_k^{(k-1)}$, $A_k^{(k)}$, $B_k^{(1)}$, $B_k^{(2)}$, $\ldots$, $B_k^{(k-1)}$, $B_k^{(k)}\}$ be a collection of $2k$ disjoint sets of vertices, where $|A_k^{(j)}| = |B_k^{(j)}| = k^{j-1}$. The vertex set of $G_k$ is the union of all sets in $H_k$. We add an edge in $G_k$ between all pairs $(u,v)$ where $u\in A_k^{(i)}$, $v\in B_k^{(j)}$ and $i \neq j$. Since $\bigcup_{j\in[k]} A_k^{(j)}$ and $\bigcup_{j\in[k]} B_k^{(j)}$ are independent, $G_k$ is bipartite. Let $n$ denote $|V(G_k)|$ and suppose that the vertices of $G_k$ are labelled arbitrarily with $[n]$. For each $i\in[k]$, let $G_k^{(i)}$ be the induced subgraph of $G_k$ on $\bigcup_{j\in[i]} (A_k^{(j)}\cup B_k^{(j)})$. We have $|V(G_k^{(i)})| = 2\sum_{j=1}^{i}k^{j-1} = 2\frac{k^i-1}{k-1} \leq \frac{2k^i}{k-1}$, and we have $n = |V(G_k^{(k)})| = 2\frac{k^k-1}{k-1} = \Theta(k^{k-1})$.

    Next, we show that the probability that $\FF$ uses $k$ colors on $G_k$ is at least $\Omega(1)$. Let $\sigma_n$ be a uniform permutation of $[n]$. For any $i \in [n]$, denote by $\sigma_n(i)$ the position of $i$ in $\sigma_n$. Let $f(A_k^{(j)}) = \min_{u\in A_k^{(j)}} \sigma_n(u)$ and define $f(B_k^{(j)})$ analogously for all $j \in [k]$. Let $x_k^{(j)} = \min(f(A_k^{(j)}),f(B_k^{(j)}))$ and $y_k^{(j)} = \max(f(A_k^{(j)}),f(B_k^{(j)}))$. Now, consider the event 
    \begin{equation}    \label{eq:E2}
        x_k^{(k)} < y_k^{(k)} < x_k^{(k-1)} < y_k^{(k-1)} < \cdots < x_k^{(1)} < y_k^{(1)}
    \end{equation}
    It is easy to verify that $\FF$ uses $k$ colors in this event. The probability of this event is at least
\begin{align*}
    &\left(\frac{|A_k^{(k)}\cup B_k^{(k)}|}{|V(G_k^{(k)})|}\right) \left(\frac{|B_k^{(k)}|}{|V(G_k^{(k)})\setminus A_k^{(k)}|}\right) \left(\frac{|A_k^{(k-1)}\cup B_k^{(k-1)}|}{|V(G_k^{(k-1)})|}\right) \left(\frac{|B_k^{(k-1)}|}{|V(G_k^{(k-1)})\setminus A_k^{(k-1)}|}\right) \\
    &\qquad\qquad\cdots\left(\frac{|A_k^{(1)}\cup B_k^{(1)}|}{|V(G_k^{(1)})|}\right) \left(\frac{|B_k^{(1)}|}{|V(G_k^{(1)})\setminus A_k^{(1)}|}\right) \\
    &= \prod_{i=0}^{k-1} \left(\frac{|A_k^{(k-i)}\cup B_k^{(k-i)}|}{|V(G_k^{(k-i)})|}\right) \left(\frac{|B_k^{(k-i)}|}{|V(G_k^{(k-i)})\setminus A_k^{(k-i)}|}\right) \\
    &\geq \prod_{i=1}^{k} \left(\frac{2k^{k-i}}{\frac{2k^{k-i+1}}{k-1}}\right)\left(\frac{k^{k-i}}{\frac{2k^{k-i+1}}{k-1}-k^{k-i}}\right).
\end{align*}
Since 
    $\frac{2k^{k-i}}{\frac{2k^{k-i+1}}{k-1}} = \frac{k-1}{k}$
and 
    $\frac{k^{k-i}}{\frac{2k^{k-i+1}}{k-1}-k^{k-i}} = \frac{k-1}{k+1}$, 
the above product is at least 
    $\prod_{i=1}^{k} \left(\frac{k-1}{k}\right)\left(\frac{k-1}{k+1}\right)
    = \left(\frac{k-1}{k}\right)^k\left(\frac{k-1}{k+1}\right)^k \ge \Omega(1)$.

    Note that the number of vertices of $G_k$ is $n = \Theta(k^{k-1})$, however, the minimum degree of $G_k$ is not $\Omega(n)$ since for example, the vertices in $A_k^{(k)}$ has degree $O(k^{k-2}) = O(n/k)$. To fix this, we add another copy of $A_k^{(k)}$ and $B_k^{(k)}$, respectively, let us denote them by $\hat{A}_k^{(k)}$ and $\hat{B}_k^{(k)}$. We connect every vertex in $\hat{A}_k^{(k)}$ to all vertices in $B_k^{(j)}$ for all $1 \le j\le k$, and  similarly connect every vertex in $\hat{B}_k^{(k)}$ to all vertices in  $A_k^{(j)}$ for all $1 \le j\le k$. Let $\hat{G}_k$ denote the new graph. By an abuse of notation we still use $n$ to denote the number of vertices of $\hat{G}_k$, which is still $n = \Theta(k^{k-1})$. Furthermore, now every vertex in $\hat{G}_k$ has minimum degree $\Omega(n)$. 
    Now, consider a similar event ${\hat{x}}_k^{(k)} < {\hat{y}}_k^{(k)} < x_k^{(k)} < y_k^{(k)} < x_k^{(k-1)} < y_k^{(k-1)} < \cdots < x_k^{(1)} < y_k^{(1)}$, where ${\hat{x}}_k^{(k)}$ and ${\hat{y}}_k^{(k)}$ are defined similarly for $\hat{A}_k^{(k)}$ and $\hat{B}_k^{(k)}$.  Clearly, this event implies $\FF$ uses color at least $k$.
    Since the event ${\hat{x}}_k^{(k)} < {\hat{y}}_k^{(k)}$ and both appear before the others happen with $\Omega(1)$ probability, and conditioned on this, the event that $x_k^{(k)} < y_k^{(k)} < x_k^{(k-1)} < y_k^{(k-1)} < \cdots < x_k^{(1)} < y_k^{(1)}$ happens in $\hat{G}_k$ is exactly the same as it happens in $G_k$.  Hence, the desired probability is $\Omega(1) \cdot \Omega(1) = \Omega(1)$. This completes the proof.
\end{proof}

\section{Limited benefit from random order}    \label{sec:less_adv}

Trivially, on complete graphs the average performance of any algorithm would be exactly the same as the worst case. 
In this section, we exhibit slightly non-trivial examples showing that on some non-complete graph classes, the random order may  provide only limited benefit over the adversarial model.

\subsection{Unit interval graphs}   \label{sec:interval}


Since an interval graph is also a perfect graph, its chromatic number equals its maximum clique size. Epstein and Levy \cite{epstein2005online} showed that for unit interval graphs, in the adversarial model, $\FF$  uses exactly $2\omega - 1$ colors, that is, $\FF$ uses at most $2\omega - 1$ colors on every unit interval graph of maximum clique size $\omega$, and there exists such interval graph on which $\FF$ uses exactly $2\omega - 1$ colors. Let $H_\omega$ denote the lower bound unit interval  graph construction from \cite{epstein2005online}, note that it is a \emph{finite} construction in the sense that the number of vertices depends only on $\omega$, say, $H_\omega$ has $n_\omega$ vertices. 

\begin{theorem} \label{thm:unit_interval_graph}
    On unit interval graphs of maximum clique size $\omega$, in the random order model, $\FF$ uses in expectation  $2\omega - 1$ colors.
\end{theorem}

\begin{proof}
    The upper bound holds because $\FF$ uses $\le 2\omega - 1$ in the adversarial model.

    For the lower bound, consider a graph $G_m$ which is $m$ independent copies of $H_\omega$. Recall $H_\omega$ has $n_\omega$ vertices. Let $E_\omega$  denote the event that in the random order  model, $\FF$ uses $2\omega - 1$ colors on $H_\omega$. By the choice of $H_\omega$, there exists a presentation order of $H_\omega$ on which $\FF$ uses $2\omega - 1$ colors, hence,
        $\Pr[E_\omega] \ge \frac{1}{n_\omega !}$.
    Let $E_m$ denote the event that in the random order model, $\FF$ uses $2\omega - 1$ colors on $G_m$, and $E^c_m$ denotes its complement event. Then, since $G_m$ consists of $m$ independent copies of $H_\omega$,
        $\Pr[E^c_m] \le (1-\frac{1}{n_\omega !})^m$.
    Because $n_\omega$  depends only on $\omega$, as $m \to \infty$, the above probability tends to $0$,  equivalently, $\Pr[E_m] \to 1$ as claimed.
\end{proof}

Clearly, the argument implies that for every graph class that admits a \emph{finite} construction for adversarial model, $\FF$ has no essential advantage in the random order model.

\subsection{High chromatic number graphs}   \label{sec:high_chi}

Frei et al \cite{ROM_treecoloring_prediction} and Bosek et al \cite{ROM_FF} showed that for trees, in the random order model, $\FF$ has average performance $O(\log n /\log\log n)$, improving its worst case performance $\Omega(\log n)$ on trees. Bosek et al also conjectures that $\FF$ has average performance $\log^{O(1)} n$ which would greatly improve its worst case performance $\Omega(n)$ on bipartite graphs. On the other hand, if we consider larger graph classes, $\FF$ can still be very non-competitive even in the random order model. In fact, Ku\v{c}era \cite{kuvcera1991greedy} already showed the following. 

\begin{theorem}[Ku\v{c}era] \label{thm:Kucera}
    For every $\epsilon > 0$, there exists a graph $G$ with $\chi(G) = O(n^\epsilon)$, but $\FF$ uses $\Omega(n/\log n)$ colors in expectation in the random order model.
\end{theorem}

We emphasize that Theorem \ref{thm:Kucera} is due to Ku\v{c}era. We only rephrase it here in the context of online graph coloring in the random order model.

\section{Open problems}

Bosek et al \cite{ROM_FF} already mentioned the open problem of understanding $\FF$'s performance on bipartite graphs and three-colorable graphs. We find the following problems  also interesting:

\begin{itemize}
    \item Characterize the graph class on which $\FF$ uses $O(1)$ colors in the random order model. 

    \item As discussed in the introduction, in the adversarial model, besides trees, $\FF$ is also optimal on inductive graphs and chordal graphs, etc.  Is $\FF$ still optimal on these graph classes in the random order model? As explained in Section \ref{sec:k-path}, the path-counting technique is not strong enough to resolve this problem. 

    \item  Another famous algorithm for online coloring of  bipartite graphs is called $\CBIP$: when a vertex $v=\sigma(i)$ arrives, $\CBIP$ computes an entire connected component $CC$ to which $v$ belongs in the partial graph known so far. If the input graph is bipartite, the connected component $CC$ can be partitioned into two sets of vertices $A$ and $B$ such that all edges go between $A$ and $B$ only. Suppose that $v \in A$, then $v$ is colored with the first color that is not present among vertices in $B$. It is known that, in the adversarial model,  $\CBIP$ also uses $\Theta(\log n)$ colors on trees just like $\FF$, Furthermore, in fact $\CBIP$ uses $\Theta(\log n)$ colors on all bipartite graphs, see \cite{lovasz1989line,bip1,bip2,li2022online}. 
    
    Does $\CBIP$ also use $\Theta(\log n / \log\log n)$ colors in expectation on trees (perhaps even bipartite graphs) in the random order model? 
\end{itemize}

\bibliography{mybib}{}

@inproceedings{ROM_FF,
  title={First-Fit Coloring of Forests in Random Arrival Model},
  author={Bosek, Bart{\l}omiej and Gutowski, Grzegorz and Laso{\'n}, Micha{\l} and Przyby{\l}o, Jakub},
  booktitle={49th International Symposium on Mathematical Foundations of Computer Science (MFCS 2024)},
  pages={33--1},
  year={2024},
  organization={Schloss Dagstuhl--Leibniz-Zentrum f{\"u}r Informatik}
}

@article{ROM_treecoloring_prediction,
  title={Tree Coloring: Random Order and Predictions},
  author={Frei, Fabian and Gehnen, Matthias and Komm, Dennis and Kr{\'a}lovi{\v{c}}, Rastislav and Kr{\'a}lovi{\v{c}}, Richard and Rossmanith, Peter and Stocker, Moritz},
  journal={arXiv preprint arXiv:2405.18151},
  year={2024}
}

@article{lovasz1989line,
  title={An on-line graph coloring algorithm with sublinear performance ratio},
  author={Lov{\'a}sz, L{\'a}szl{\'o} and Saks, Michael and Trotter, William T},
  journal={Discrete Mathematics},
  volume={75},
  number={1-3},
  pages={319--325},
  year={1989},
  publisher={Elsevier}
}

@article{bip1,
  title={Online coloring of bipartite graphs with and without advice},
  author={Bianchi, Maria Paola and B{\"o}ckenhauer, Hans-Joachim and Hromkovi{\v{c}}, Juraj and Keller, Lucia},
  journal={Algorithmica},
  volume={70},
  number={1},
  pages={92--111},
  year={2014},
  publisher={Springer}
}

@article{albers2021tight,
  title={Tight bounds for online coloring of basic graph classes},
  author={Albers, Susanne and Schraink, Sebastian},
  journal={Algorithmica},
  volume={83},
  number={1},
  pages={337--360},
  year={2021},
  publisher={Springer}
}

@inproceedings{irani1990coloring,
  title={Coloring inductive graphs on-line},
  author={Irani, S},
  booktitle={Proceedings [1990] 31st Annual Symposium on Foundations of Computer Science},
  pages={470--479},
  year={1990},
  organization={IEEE}
}

@inproceedings{bip2,
  title={Lower bounds for on-line graph colorings},
  author={Gutowski, Grzegorz and Kozik, Jakub and Micek, Piotr and Zhu, Xuding},
  booktitle={International Symposium on Algorithms and Computation},
  pages={507--515},
  year={2014},
  organization={Springer}
}

@article{li2023online,
  title={Online vector bin packing and hypergraph coloring illuminated: Simpler proofs and new connections},
  author={Li, Yaqiao and Pankratov, Denis},
  journal={Procedia Computer Science},
  volume={223},
  pages={70--77},
  year={2023},
  publisher={Elsevier}
}

@article{kuvcera1991greedy,
  title={The greedy coloring is a bad probabilistic algorithm},
  author={Ku{\v{c}}era, Lud{\v{e}}k},
  journal={Journal of Algorithms},
  volume={12},
  number={4},
  pages={674--684},
  year={1991},
  publisher={Elsevier}
}

@inproceedings{epstein2005online,
  title={Online interval coloring and variants},
  author={Epstein, Leah and Levy, Meital},
  booktitle={International Colloquium on Automata, Languages, and Programming},
  pages={602--613},
  year={2005},
  organization={Springer}
}

@article{li2022online,
  title={Online coloring and a new type of adversary for online graph problems},
  author={Li, Yaqiao and Narayan, Vishnu V and Pankratov, Denis},
  journal={Algorithmica},
  volume={84},
  number={5},
  pages={1232--1251},
  year={2022},
  publisher={Springer}
}

@article{hyperC_1,
  title={Online hypergraph coloring},
  author={Nagy-Gy{\"o}rgy, Judit and Imreh, Cs},
  journal={Information Processing Letters},
  volume={109},
  number={1},
  pages={23--26},
  year={2008},
  publisher={Elsevier}
}

@article{graph_gyarfas1988line,
  title={On-line and first fit colorings of graphs},
  author={Gy{\'a}rf{\'a}s, Andr{\'a}s and Lehel, Jen{\"o}},
  journal={Journal of Graph theory},
  volume={12},
  number={2},
  pages={217--227},
  year={1988},
  publisher={Wiley Online Library}
}
\bibliographystyle{alpha}

\end{document}